\documentclass[english,11pt]{article}

\usepackage{comment}
\usepackage[english]{babel}
\usepackage{inputenc}
\usepackage[T1]{fontenc}
\usepackage{geometry}
\usepackage{xcolor}
\usepackage{enumerate}
\usepackage{amsmath,amsthm,amssymb,amsfonts,bbm,bm,nicefrac,mathtools}
\usepackage{hyperref}
\usepackage[normalem]{ulem}

\newcommand{\R}{\mathbb{R}}

\newcommand{\N}{\mathbb{N}}

\newcommand{\dH}{\Delta H}

\newcommand{\tsigma}{\tilde{\sigma}}

\newcommand{\E}{\mathbb{E}}

\newcommand{\proba}{\mathbb{P}}

\newcommand{\thermal}[1]{\left\langle#1\right\rangle}
\newcommand{\Var}[1]{\mbox{Var}\left(#1\right)}

\newcommand{\tJ}{\tilde{J}}

\newcommand{\tilh}{\tilde{h}}

\newtheorem{theorem}{Theorem}
\theoremstyle{plain}
\newtheorem{proposition}{Proposition}
\newtheorem{lemma}{Lemma}
\newtheorem{corollary}{Corollary}
\theoremstyle{plain}

\theoremstyle{plain}

\theoremstyle{plain}
\newtheorem{definition}{Definition}
\theoremstyle{remark}
\newtheorem*{remark}{Remark}
\theoremstyle{definition}

\definecolor{ps}{RGB}{0,0,200}

\title{A contiguity approach to replica symmetric marginals}

\author{Ernesto Mordecki, Anas Rahman, and Manuel S\'aenz}

\begin{document}

\maketitle

\begin{abstract}

    We develop a probabilistic cavity-contiguity framework for proving replica-symmetric convergence of local marginals in mean-field Gibbs systems. The approach is based on cavity decompositions of the Hamiltonian together with direct comparison of probability measures through Radon-Nikodym derivatives and Hellinger-type estimates. At a conceptual level, the method separates the concentration of the relevant order parameters from the identification of the asymptotic cavity model and the comparison of the associated Gibbs measures. In contrast with interpolation-based approaches, the argument relies only weakly on the Gaussianity of the disorder and naturally accommodates concentration tools such as Poincaré and log-Sobolev inequalities. Rather than pursuing maximal generality, with the aim of making the method transparent, we implement the framework in a canonical example: the high-temperature Sherrington-Kirkpatrick model. In this setting, we prove that the marginal law of a fixed spin converges in total variation toward the effective one-dimensional cavity measure predicted by the replica method. Beyond the specific result for the Sherrington-Kirkpatrick model, the paper illustrates a broader cavity-contiguity methodology which is expected to extend naturally to other mean-field Gibbs systems, particularly Bayesian inference models with or without mismatch.

\end{abstract}


\section{Introduction}\label{sec:intro}

A recurring phenomenon in high-dimensional mean-field systems is the emergence of effective low-dimensional descriptions for local observables in the thermodynamic limit. In spin glasses, Bayesian inference, and high-dimensional statistics, methods from statistical physics predict that finite collections of coordinates asymptotically decouple and behave according to effective scalar Gibbs measures driven by random cavity fields \cite{mezard2009information,mezard1987spin}. The approach described here draws heavily from the heuristic cavity method of these references. These asymptotic scalar descriptions play a central role in the modern theory of spin glasses \cite{panchenko2013sherrington,talagrand2010mean}, in the analysis of TAP equations \cite{fan2021tap}, and in algorithmic frameworks such as approximate message passing and state evolution \cite{zou2022concise}. Closely related asymptotic decoupling phenomena also appear throughout contemporary high-dimensional inference and machine learning \cite{krzakala2024statistical,zdeborova2016statistical}, where Bayesian posteriors and Gibbs measures often admit effective one-dimensional descriptions in the large-system limit.

Questions concerning effective scalar descriptions, asymptotic decoupling, and the probabilistic structure of high-dimensional Gibbs measures have become increasingly central in modern high-dimensional statistics, Bayesian inference, and statistical physics; see, for example, the recent survey \cite{maleki2026high} and the references therein. In Bayes-optimal settings, these phenomena are closely related to Nishimori identities and replica-symmetric structures \cite{contucci2009spin,nishimori2001statistical}, and have led to a detailed understanding of several canonical inference models \cite{barbier2016mutual,barbier2019optimal,barbier2020mutual}. More broadly, Bayesian methods have become a central tool in modern high-dimensional statistics; see \cite{banerjee2021bayesian} for a recent overview.

The rigorous analysis of these local asymptotic descriptions is often technically delicate. Existing approaches typically rely on interpolation methods \cite{barbier2019adaptive,panchenko2013sherrington,talagrand2010mean}, asymptotic descriptions of message passing algorithms \cite{zou2022concise}, TAP equations \cite{fan2021tap}, or detailed analyses of free energies and overlap structures \cite{panchenko2013sherrington,talagrand2010mean}. In Bayesian inference settings, related asymptotic scalar descriptions have also been obtained through leave-one-out and cavity-based approaches \cite{barbier2025performance,barbier2022marginals,saenz2025characterizing}. In particular, the specific combination of a leave-one-out cavity decomposition with a Radon-Nikodym tilt that motivates the present work first appeared in the companion paper \cite{saenz2025characterizing}, in the context of high-dimensional Bayesian generalized linear models. There, the full posterior is represented as a tilt of leave-a-variable-out and leave-an-observation-out posteriors; see, in particular, Theorem 2 of that work. The corresponding tilt is shown to admit an asymptotically Gaussian description, providing a high-dimensional analogue of classical local asymptotic analysis.

From the perspective of classical asymptotic statistics, this construction is naturally connected with Le Cam's theory, in which contiguity and changes of measure through likelihood-ratio tilts play a central role \cite{LeCamYang1990,van2000asymptotic}. During the development of the companion GLM paper \cite{saenz2025characterizing}, P.S. observed that the same change-of-measure mechanism applies beyond generalized linear models to a broader family of mean-field Gibbs systems. In the models considered here, this mechanism can be formulated particularly transparently, because the cavity decomposition, the concentration of the relevant order parameters, and the comparison of probability measures can be treated as separate components.

The purpose of the present paper is therefore not to introduce the leave-one-out tilting idea itself, but to isolate its underlying probabilistic mechanism and develop it as a comparatively direct and modular route to local replica-symmetric convergence in mean-field Gibbs systems. The central perspective is that replica-symmetric local convergence can be formulated as a quantitative contiguity problem between the full Gibbs measure and an associated cavity measure. Rather than relying primarily on interpolation identities, the approach developed here is based on cavity decompositions of the Hamiltonian, direct comparison of probability measures through Radon-Nikodym derivatives, and quantitative control of Hellinger-type distances in the spirit of asymptotic statistical theory \cite{LeCamYang1990,van2000asymptotic}.

At a conceptual level, the method separates the proof into three comparatively independent components: concentration of the relevant order parameters, identification of the asymptotic cavity fields, and quantitative comparison between the finite-dimensional and limiting Gibbs measures. One of the advantages of this decomposition is that it isolates the model-dependent concentration problem from the probabilistic mechanism governing the emergence of the asymptotic scalar channel. As a consequence, the framework can naturally leverage tools from concentration of measure theory \cite{boucheron2003concentration}, such as Poincaré inequalities, log-Sobolev inequalities, transportation inequalities, or related functional inequalities, without substantially modifying the remaining structure of the argument. In the particular setting analysed in this paper, the required concentration estimates rely on the log-Sobolev inequalities established in \cite{bauerschmidt2019very}.

A second advantage of the method is that it relies only weakly on the Gaussianity of the disorder. In the Sherrington-Kirkpatrick model analysed here, Gaussianity enters mainly through the concentration estimates and the asymptotic description of the cavity fields. In contrast with interpolation-based approaches, where Gaussian integration by parts plays a central structural role in the derivation of the fundamental identities, the cavity-contiguity framework developed in this paper is largely measure-theoretic and probabilistic in nature. This suggests that the strategy should be comparatively robust with respect to modifications of the disorder distribution and potentially compatible with broader families of distributions.

The argument also possesses a standard probabilistic structure which is naturally compatible with stochastic-process formulations. Since the proof is based on comparison of probability measures through Radon-Nikodym derivatives and Hellinger-type quantities, the same general strategy can in principle be adapted to path-space settings through the corresponding notion of Hellinger processes \cite{jacod2013limit}. In particular, the framework appears well-suited for the study of asymptotic scalar descriptions associated with Langevin dynamics and related stochastic evolutions, which is the subject of ongoing work.

The main contribution of the paper is therefore methodological rather than purely model-specific. We do not attempt to treat the largest possible class of models in full generality. Instead, we formulate the cavity-contiguity framework at an abstract level and then implement it completely in a canonical example: the high-temperature Sherrington-Kirkpatrick model. This setting is sufficiently rich to exhibit all the main ingredients of the argument while remaining explicit enough that the probabilistic mechanism underlying replica-symmetric local convergence can be analysed transparently. One of the strengths of the approach is precisely that the resulting proofs are comparatively short, modular, and probabilistically explicit.

In this sense, \cite{saenz2025characterizing} and the present article should be viewed as companion papers. The former contains the first model-specific appearance of the leave-one-out tilting mechanism in a high-dimensional Bayesian generalized linear model and identifies its connection with classical local asymptotic analysis. The present paper isolates the underlying cavity-contiguity mechanism, formulates it as a modular proof strategy, and implements it in a canonical spin-glass model. The two works therefore emphasise complementary aspects of the same probabilistic perspective.

The main limitation of the present approach is that the quantitative bounds obtained through the current implementation are not expected to yield optimal convergence rates. The framework prioritises structural transparency, modularity, and probabilistic robustness over sharp quantitative optimisation. Moreover, although the method separates the concentration problem from the asymptotic identification step, sufficiently strong concentration estimates for the relevant order parameters remain essential for the argument to close.

Beyond the Sherrington-Kirkpatrick model itself, the structure of the argument suggests possible extensions to a broader class of mean-field Gibbs systems, particularly Bayesian inference models on the Nishimori line \cite{barbier2022strong,nishimori2001statistical,zdeborova2016statistical} or with mismatch \cite{barbier2025performance,saenz2025characterizing}. More broadly, the framework developed here is closely connected to the modern perspective that high-dimensional Gibbs systems often admit asymptotically decoupled low-dimensional descriptions.

\paragraph{Organisation of the paper.} In Section \ref{sec:general-strategy}, we describe the general cavity-contiguity framework for proving replica-symmetric convergence of local marginals. Section \ref{sec:example} introduces the high-temperature Sherrington-Kirkpatrick model, develops the associated cavity decomposition, and implements the general strategy in this setting. Some technical proofs are deferred to Appendices \ref{app:cholesky} and \ref{app:proof_applications}.


\section{Replica-symmetric marginals through contiguity}\label{sec:general-strategy}

In this section we describe the strategy at an informal level, with the aim of separating the structural probabilistic mechanism from the model-dependent estimates required to implement it. In Section \ref{sec:example}, we implement, for the high-temperature Sherrington-Kirkpatrick model, the strategy described here.

Consider a sequence of Gibbs measures on product spaces of the form
\begin{equation*}
    \proba_N(d\sigma\mid D) = \frac{1}{Z_N(D)}\exp\{H_N(\sigma,D)\}\prod_{i=1}^N\mu(d\sigma_i),
\end{equation*}
where $D$ denotes the disorder of the model. In the SK model, $D$ is a symmetric matrix of Gaussian couplings, while in Bayesian inference models it may contain a collection of observations. For some fixed $i_0\geq1$, the goal is to understand the asymptotic behaviour of the $i_0$-th marginal $\proba_{N,i_0}(\cdot\mid D)$. In particular, we will consider the Total Variation norm, defined below, to establish their convergence.

\begin{definition}
    If $\mu$ is a signed measure on $(\Omega,\mathcal{F})$, we define its total variation norm according to
    \begin{equation*}
        \|\mu\|_{\rm TV} = \sup\left\{ \int f  d\mu : f \ \text{is real valued and $\mathcal{F}$-measurable with } |f| \leq 1 \right\}.
    \end{equation*}
\end{definition}

The first step is to establish a cavity decomposition of the Hamiltonian. Namely, after removing a coordinate $i_0$, one seeks an identity of the form
\begin{equation}\label{eq:H-decomposition}
    H_N(\sigma,D) = H_{N-1}(\tsigma,\widetilde D) + \Delta H_N(\sigma_{i_0},\tsigma,D_{i_0}) + \delta_N(\sigma,D),
\end{equation}
where $\tsigma$ denotes the configuration obtained from $\sigma$ by removing its $i_0$-th coordinate, $\Delta H_N$ contains the leading interaction between the removed coordinate and the cavity system, and $\delta_N$ is an error term. Here we assume that the disorder can be decomposed as $D = (\widetilde D, D_{i_0})$ in accordance with \eqref{eq:H-decomposition}: $\widetilde D$ are the elements of $D$ appearing in the term $H_{N-1}$ and $D_{i_0}$ those appearing in $\Delta H_N$.

Let $\proba_N'(\cdot\mid D)$ be the Gibbs measure associated with the Hamiltonian $H_N-\delta_N$. One then needs to show that the correction term $\delta_N$ is negligible at the level of probability measures. A convenient sufficient condition is
\begin{equation*}
    \varepsilon_N := \E\thermal{\delta_N}-\E\thermal{\delta_N}' \longrightarrow 0,
\end{equation*}
where $\thermal{\cdot}$ and $\thermal{\cdot}'$ denote expectations under $\proba_N$ and $\proba_N'$, respectively. By Pinsker's inequality and Jensen's inequality, this implies
\begin{equation*}
    \E\|\proba_N(\cdot\mid D)-\proba_N'(\cdot\mid D)\|_{\rm TV}^2 \leq \frac{\varepsilon_N}{2}.
\end{equation*}

The second step is to identify the relevant order parameters and prove their concentration. These quantities are model-dependent and typically include overlaps, magnetisations, empirical moments, or correlations with a planted signal. Denoting them collectively by $\Gamma_N$, the required statement takes the form
\begin{equation*}
    \E\thermal{\|\Gamma_N-\gamma\|^2}_c \longrightarrow 0
\end{equation*}
for some deterministic limit $\gamma$. Here $\thermal{\cdot}_c$ denotes expectation under the cavity measure. Proving this concentration is usually the most model-specific part of the strategy and may require functional inequalities, Nishimori identities, or other concentration tools.

The cavity measure itself is defined from the cavity Hamiltonian according to
\begin{equation*}
    \proba_{N,c}(d\sigma\mid \widetilde D) = \frac{1}{Z_{N,c}}\exp\{H_{N-1}(\tsigma,\widetilde D)\}\prod_{i=1}^N\mu(d\sigma_i).
\end{equation*}
Under this measure, the coordinate $\sigma_{i_0}$ is independent of the cavity system and distributed according to the prior measure $\mu$.
The cavity decomposition allows one to compare directly the measures $\proba_N'$ and $\proba_{N,c}$. More precisely, a direct computation gives
\begin{equation*}
    \frac{d\proba_N'}{d\proba_{N,c}}(\sigma) = \frac{\exp\{\Delta H_N(\sigma_{i_0},\tsigma,D_{i_0})\}}{\thermal{\exp\{\Delta H_N(\sigma_{i_0},\tsigma,D_{i_0})\}}_c}.
\end{equation*}
Projecting this identity onto the coordinate $i_0$ yields the Radon-Nikodym derivative of the marginal law $\proba_{N,i_0}'$ with respect to the prior measure $\mu$. Denoting this derivative by $\mathcal{D}_{N,i_0}$, one obtains
\begin{equation*}
    \mathcal{D}_{N,i_0}(s) = \frac{\int \exp\{H_{N-1}(\tsigma,\tilde{D}) + \Delta H_N(s,\tsigma,D_{i_0})\} \prod_{i\backslash\{i_0\}}\mu(d\sigma_i)}{\thermal{\exp\{\Delta H_N(\sigma_{i_0},\tsigma,D_{i_0})\}}_c},
\end{equation*}
for $s$ in the support of $\mu$.

The next step is to identify a limiting random density $\mathcal{D}_{i_0}$. This typically arises from the convergence of the cavity fields appearing in $\Delta H_N$. In many situations, after conditioning on the order parameters, these fields become asymptotically Gaussian. The limiting cavity law $\mu_{i_0}$ is then defined through
\begin{equation*}
    \frac{d\mu_{i_0}}{d\mu}(s) = \mathcal{D}_{i_0}(s).
\end{equation*}

To compare the finite-$N$ and limiting marginals, a natural quantity to study is the Hellinger integral, defined below.

\begin{definition}\label{def:hell-int}
    Let $(\Omega, \mathcal{F})$ be a measurable space endowed with three probability measures $\proba$, $\proba'$, and $\mathbb{Q}$ such that $\proba \ll \mathbb{Q}$ and $\proba' \ll \mathbb{Q}$; and consider the Radon–Nikodym derivatives $d\proba/d\mathbb{Q}$ and $d\proba'/d\mathbb{Q}$. The Hellinger integral of $\proba$ and $\proba'$ is defined according to
    \begin{equation*}
        I(\proba,\proba') = \E_\mathbb{Q}\left( \sqrt{\frac{d\proba}{d\mathbb{Q}} \cdot \frac{d\proba'}{d\mathbb{Q}}} \right).
    \end{equation*}
\end{definition}

\begin{remark}
    It can be proved that the value of the Hellinger integral does not depend on the background measure $\mathbb{Q}$ used to define it.
\end{remark}

Since one can easily see that
\begin{equation*}
    \|\proba_{N,i_0}'-\mu_{i_0}\|_{\rm TV}^2 \leq 2\big(1-I(\proba'_{N,i_0},\mu_{i_0})\big),
\end{equation*}
for our purposes, it suffices to prove that $I(\proba'_{N,i_0},\mu_{i_0})\to1$ in expectation. However, to be able to define the Hellinger integral, one needs to place the finite-$N$ cavity fields and their limiting counterparts on a common probability space. This is achieved through an explicit coupling construction, typically obtained by coupling the covariance structure of the finite-dimensional cavity fields with the limiting covariance prescribed by the order parameters. In the SK model, this is done through a coupling of Gaussian vectors whose covariance matrices converge because of overlap concentration. This is explained in more detail in Section \ref{sec:coupling}.

The final step, after establishing the above coupling, is to prove that
\begin{equation*}
    \E I(\proba'_{N,i_0},\mu_{i_0})\longrightarrow1.
\end{equation*}
At a heuristic level, this follows from the concentration of the order parameters together with continuity properties of the cavity densities as functions of the cavity fields. In practice, the proof typically proceeds by establishing quantitative control of logarithmic moments of the Radon-Nikodym derivatives and then replacing the finite-$N$ cavity fields by their limiting counterparts using the coupling construction.

As presented in Section \ref{sec:example}, all these steps can be carried out explicitly in the high-temperature Sherrington-Kirkpatrick. For other mean-field models, such as Bayesian rank-one inference problems, the same structure remains available, although the identification and concentration of the relevant order parameters may involve substantial adaptations of the proof.


\section{Example: high-temperature Sherrington-Kirkpatrick model}\label{sec:example}

\subsection{Setting and notation}

For each $N\geq1$, let $J=(J_{ij})_{1\leq i,j\leq N}\in\R^{N \times N}$ be a Wigner matrix. That is, $J$ is symmetric and its entries satisfy
\begin{equation*}
    J_{ij} = \begin{cases}
            \frac{1}{\sqrt{N}}z_{ij} & \text{if } i < j \\
            \sqrt{\frac{2}{N}}\,z_{ii} & \text{if } i=j
        \end{cases}
\end{equation*}
with $(z_{ij})_{1 \leq i \leq j \leq N }$ independent standard Gaussian random variables.

We consider the Sherrington-Kirkpatrick Hamiltonian
\begin{equation*}
    H_N(\sigma) = \beta \sum_{i,j=1}^N J_{ij}\sigma_i\sigma_j + h\sum_{i=1}^N \sigma_i,
\end{equation*}
defined on spin configurations $\sigma=(\sigma_1,\dots,\sigma_N)\in\{-1,1\}^N$, where $\beta>0$ is the inverse temperature and $h\neq0$ is the external magnetic field. 

\begin{remark}
    Here we define the double sum in the Hamiltonian to include all the crossed terms, not just the ones above the diagonal. This does not change much in the proofs but should be kept in mind that it does change the interpretation of the inverse temperature parameter by $1/2$. Indeed, our high-temperature region $\beta < 1/4$ is the same as the one $\beta < 1/2$ present in \cite{talagrand2010mean}.
\end{remark}

The associated Gibbs measure on $\{-1,1\}^N$ is given by
\begin{equation*}
    \proba_N(\sigma\mid J) = \frac{1}{Z_N(J)}\exp\{H_N(\sigma,J)\},
\end{equation*}
where $Z_N(J) > 0$ is the partition function. Expectation with respect to the Gibbs measure will be denoted by $\thermal{\cdot}$. Expectation and variance with respect to the disorder $J$ will be denoted by $\E(\cdot)$ and $\Var{\cdot}$, respectively.

Throughout the paper, $(\sigma^{(l)})_{l\geq1}$ will denote independent samples, or replicas, drawn from the Gibbs measure $\proba_N(\cdot\mid J)$. Their overlaps are defined by
\begin{equation*}
    Q_{l,l'} := \frac{1}{N}\sum_{i=1}^N \sigma_i^{(l)}\sigma_i^{(l')}.
\end{equation*}
In the high-temperature regime considered in this paper, it is known that the overlap concentrates around the unique solution of the replica-symmetric fixed point equation
\begin{equation}\label{eq:fpe}
    q = \E_z \tanh^2(2\beta\sqrt{q}\,z+h),
\end{equation}
where $z$ is a standard Gaussian random variable.


\subsection{Cavity decomposition and replica symmetry}

Fix $i_0\in[N]$. Given a configuration $\sigma\in\{-1,1\}^N$, let
\begin{equation*}
    \tsigma = (\sigma_1,\dots,\sigma_{i_0-1},\sigma_{i_0+1},\dots,\sigma_N)
\end{equation*}
be the vector obtained from $\sigma$ by removing its $i_0$-th coordinate. Similarly, let $\tJ\in\R^{(N-1)\times(N-1)}$ be the matrix obtained from $J$ by deleting its $i_0$-th row and column, and let
\begin{equation*}
    J_{i_0} = (J_{1i_0},\dots,J_{i_0-1,i_0},J_{i_0+1,i_0},\dots,J_{Ni_0})\in\R^{N-1}
\end{equation*}
be the vector formed by the $i_0$-th row of $J$ without its diagonal entry.

The associated cavity Hamiltonian is defined by
\begin{equation*}
    H_{N-1}(\tsigma,\tJ) = \beta \sum_{i,j\neq i_0} J_{ij}\sigma_i\sigma_j + h\sum_{i\neq i_0}\sigma_i.
\end{equation*}
The corresponding cavity Gibbs measure on $\{-1,1\}^N$ is
\begin{equation*}
    \proba_{N,c}(\sigma\mid \tJ) = \frac{1}{Z_{N,c}(\tJ)}\exp\{H_{N-1}(\tsigma,\tJ)\},
\end{equation*}
where $Z_{N,c}>0$ is the associated partition function.

Under the cavity measure, the spin $\sigma_{i_0}$ is independent of $\tsigma$ and distributed according to the prior measure $\mu$. Expectation with respect to $\proba_{N,c}$ will be denoted by $\thermal{\cdot}_c$. The following lemma gives the cavity decomposition of the SK Hamiltonian. For this, let
\begin{equation*}
    \Theta_N := J_{i_0}^{\top}\tsigma = \sum_{i\neq i_0}J_{ii_0}\sigma_i = \frac{1}{\sqrt{N}}\sum_{i\neq i_0} z_{ii_0}\sigma_i.
\end{equation*}
In physical terms, this is the cavity field acting on the coordinate $i_0$.

\begin{lemma}\label{lem:cavity-decomposition}
    For every $i_0\in[N]$, we have
    \begin{equation*}
        H_N(\sigma,J) = H_{N-1}(\tsigma,\tJ) + \Delta H_N(\sigma_{i_0},\Theta_N) + \delta_N(J);
    \end{equation*}
    where $\Delta H_N(\sigma_{i_0},\Theta_N) :=  2\beta \sigma_{i_0}\Theta_N + h\sigma_{i_0}$ and $\delta_N(J) = \beta J_{i_0i_0}$.
\end{lemma}
\begin{proof}
    Decompose the interaction term according to whether indices are equal to $i_0$ or not. Since $J$ is symmetric,
    \begin{equation*}
        \sum_{i,j=1}^N J_{ij}\sigma_i\sigma_j = \sum_{i,j\neq i_0} J_{ij}\sigma_i\sigma_j + 2\sigma_{i_0}\sum_{i\neq i_0}J_{ii_0}\sigma_i + J_{i_0i_0}\sigma_{i_0}^2.
    \end{equation*}
    Because $\sigma_{i_0}^2=1$, the result follows immediately.
\end{proof}

\begin{remark}
    In the SK model, $\delta_N$ does not depend on $\sigma$ and can then be absorbed into the normalisation constant. In general, this error terms should be handled as described in Section \ref{sec:general-strategy}.
\end{remark}

The cavity decomposition shows that the local behaviour of the spin $\sigma_{i_0}$ is governed by the cavity field $\Theta_N$. Under the cavity measure, conditional on the replicas $(\tsigma^{(l)})_{l\geq1}$, the family $(\Theta_N^{(l)})_{l\geq1}$ is Gaussian with covariance structure determined by the overlaps
\begin{equation*}
    Q_{l,l'} := \frac{1}{N}\sum_{i\neq i_0}\sigma_i^{(l)}\sigma_i^{(l')}.
\end{equation*}

The replica-symmetric regime corresponds to concentration of these overlaps around a deterministic value. The following proposition establishes overlap concentration in the high-temperature regime. This is a standard result which can be found, for example, in \cite{talagrand2010mean}. Here we present a self-contained proof which, in contrast to that in \cite{talagrand2010mean}, does not use interpolation techniques.

\begin{proposition}\label{prop:rs}
    Assume that $\beta<1/4$ and $h\neq0$. Then there exists a unique solution $q\in[0,1]$ of \eqref{eq:fpe} and, for some fixed $C > 0$, we have that
    \begin{equation*}
        \E\thermal{(Q_{1,2}-q)^2}_c \leq \frac{C}{N}
    \end{equation*}
\end{proposition}

\begin{remark}
    We will prove this proposition making use of the main result of \cite{bauerschmidt2019very}, a log-Sobolev inequality valid in the high-temperature regime. The resulting proof is quite concise and, as far as we know, new. 
\end{remark}

The proof of Proposition \ref{prop:rs} is deferred to Appendix \ref{app:proof_applications}. Throughout the paper, the value $q$ will denote the unique replica-symmetric solution appearing in Proposition \ref{prop:rs}.


\subsection{Convergence of marginals}\label{sec:main}

We now state the main result of the paper. Fix $i_0\in[N]$ and let $z$ be a standard Gaussian random variable. Define the random probability measure $\mu_{i_0}$ on $\{-1,1\}$ by
\begin{equation}\label{eq:limit-marginal}
    \mu_{i_0}(\sigma) = \frac{\exp\{(2\beta\sqrt{q}\,z+h)\sigma\}}{2\cosh(2\beta\sqrt{q}\,z+h)}.
\end{equation}
As explained in Section \ref{sec:general-strategy}, this asymptotical marginal is obtained by deriving the asymptotic derivative between the replica and full measures.

The following proposition establishes convergence of the Gibbs marginal of a fixed spin toward this replica-symmetric cavity law when the overlaps concentrate towards a fixed constant.

\begin{proposition}\label{prop:almost-main}
    Let $\proba_{N,i_0}(\cdot\mid J)$ denote the marginal law of $\sigma_{i_0}$ under the Gibbs measure $\proba_N(\cdot\mid J)$ and assume that there are constants $K > 0$ and $q \in[-1,1]$ such that $\E\thermal{(Q_{12}-q)^2}_c\leq K/N$. Then, there is some constant $C>0$ such that
    \begin{equation*}
        \E\|\proba_{N,i_0}(\cdot\mid J)-\mu_{i_0}(\cdot)\|_{\rm TV}^2 \leq \frac{C}{N^{1/8}}.
    \end{equation*}
\end{proposition}

Finally, the main theorem of the section is the validity of the conclusion of the above proposition in the high-temperature regime.

\begin{theorem}\label{thm:main_sk}
    Let $\proba_{N,i_0}(\cdot\mid J)$ be as in the above proposition and assume that $\beta<1/4$ and $h\neq0$. Then, there is some constant $C>0$ such that
    \begin{equation*}
        \E\|\proba_{N,i_0}(\cdot\mid J)-\mu_{i_0}(\cdot)\|_{\rm TV}^2 \leq \frac{C}{N^{1/8}}.
    \end{equation*}
\end{theorem}

\begin{remark}
    Notice that a version of this Theorem is already present in \cite{talagrand2010mean}. The main differences are that, the version therein has optimal convergence rates. As discussed above, the interest of this theorem does not lie in the statement itself, which is classical, but in the probabilistic mechanism used in its proof.
\end{remark}

The theorem shows that, in the replica-symmetric regime, the local behaviour of the SK model is asymptotically described by a one-dimensional effective Gibbs measure driven by a Gaussian cavity field. In particular, the marginal distribution of a single spin becomes asymptotically equivalent, in total variation, to the cavity prediction arising from the replica and TAP heuristics.

As an immediate consequence, one obtains convergence of averages of local observables.

\begin{corollary}\label{cor:local-observables}
    Assume that $\beta<1/4$ and $h\neq0$. Let $f:\{-1,1\}^2\to\R$ be measurable. Then
    \begin{equation*}
        \lim_{N\to\infty}\frac{1}{N}\sum_{i=1}^N \E\thermal{f(\sigma^{(1)}_i,\sigma^{(2)}_i)} = \E f(\sigma_{\star},\sigma'_{\star}),
    \end{equation*}
    where $\sigma_\star,\sigma'_{\star}\in\{-1,1\}$ are two conditionally independent random variables distributed according to
    \begin{equation*}
        \proba(\sigma_\star=\sigma\mid z) = \frac{\exp\{(2\beta\sqrt{q}\,z+h)\sigma\}}{2\cosh(2\beta\sqrt{q}\,z+h)}.
    \end{equation*}
\end{corollary}
\begin{proof}
    For a function of only one replica we have, by exchangeability,
    \begin{equation*}
        \frac{1}{N}\sum_{i=1}^N \E\thermal{f(\sigma_i)} = \E\thermal{f(\sigma_1)}
    \end{equation*}
    and the result then follows directly from Theorem \ref{thm:main_sk}. For functions of two replicas it follows in a similar way by noticing that the proof of the theorem works in the same way for any finite number of replicas.
\end{proof}

In particular, this last corollary proves that the limit $q$ of the overlaps $Q_{ll'}$ is a solution of the fixed point equation \eqref{eq:fpe}.

Define the free-energy of the Sherrington-Kirkpatrick model according to
\begin{equation*}
    F_N(\beta,h):=\frac{1}{N}\E\log Z_N(J).
\end{equation*}
As shown by the following corollary, Theorem \ref{thm:main_sk} also implies a characterisation of the asymptotic free-energy.

\begin{corollary}\label{cor:free-energy}
    Assume that $\beta<1/4$ and $h\neq0$ and let $q$ be the unique solution of \eqref{eq:fpe}. Then
    \begin{equation*}
        \lim_{N\to\infty}F_N(\beta,h)=\E\log 2\cosh(2\beta\sqrt q Z+h)+\beta^2(1-q)^2,
    \end{equation*}
    with $Z$ a standard Gaussian random variable.
\end{corollary}
\begin{proof}
    Let $q(s)$ be the unique solution of
    \begin{equation*}
        q(s)=\E\tanh^2(2s\sqrt{q(s)}Z+h).
    \end{equation*}
    By Gaussian integration by parts, uniformly on compact subsets of the high-temperature region, we can see that
    \begin{equation*}
        \partial_\beta F_N(\beta,h)=2\beta\left(1-\E\thermal{Q_{12}^2}\right)+o_N(1).
    \end{equation*}
    By Proposition \ref{prop:rs} and Corollary \ref{cor:local-observables}, we have the $\E\thermal{Q_{12}^2}$ converges to $q(\beta)^2$. Therefore, $\partial_\beta F_N(\beta,h)$ converges towards $2\beta(1-q(\beta)^2)$ within the high-temperature region. Since $F_N(0,h)=\log(2\cosh h)$, integration in $\beta$ gives
    \begin{equation}\label{eq:equivalent-fe}
        \lim_{N\to\infty}F_N(\beta,h)=\log(2\cosh h)+\int_0^\beta 2s(1-q(s)^2)\,ds.
    \end{equation}
    Here we can exchange limit and integral because all the quantities involved are bounded. This establishes the limit of the free energy. To prove the equivalent variational expression in the statement of the corollary, define
    \begin{equation*}
        \Phi(\beta,q):=\E\log 2\cosh(2\beta\sqrt q Z+h)+\beta^2(1-q)^2.
    \end{equation*}
    The final formula can be seen to be equivalent to \eqref{eq:equivalent-fe} by differentiating along the fixed point $q=q(\beta)$ given by the unique solution of \eqref{eq:fpe} for $q(\beta)$. Indeed, if we compute the partial derivative of $\Phi$ with respect to $\beta$, keeping $q$ fixed, we get
    \begin{equation*}
        \partial_\beta \Phi(\beta,q)=2\sqrt q\,\E\left[Z\tanh(2\beta\sqrt q Z+h)\right]+2\beta(1-q)^2.
    \end{equation*}
    By Gaussian integration by parts and standard identities for hyperbolic functions, this is the same as
    \begin{equation*}
        \partial_\beta \Phi(\beta,q)=4\beta q\left(1-\E\tanh^2(2\beta\sqrt q Z+h)\right)+2\beta(1-q)^2.
    \end{equation*}
    In a similar way, if we compute the partial derivative with respect to $q$ and use Gaussian integration by parts, we get
    \begin{equation*}
        \partial_q \Phi(\beta,q)=2\beta^2\left(q-\E\tanh^2(2\beta\sqrt q Z+h)\right).
    \end{equation*}
    Therefore, by \eqref{eq:fpe}, we have $\partial_q\Phi(\beta,q)=0$ and $\partial_\beta\Phi(\beta,q)=2\beta(1-q^2)$. We then have that the derivatives with respect to $\beta$ of \eqref{eq:equivalent-fe} and the final formula coincide. And because the values at $\beta=0$ also match, both expressions are equivalent.

\end{proof}


\subsection{Proof of Theorem \ref{thm:main_sk}}\label{sec:proofs}

\subsubsection{Coupling of probability measures}\label{sec:coupling}

Consider the array $(\Theta_N^{(l)})_{l\geq1}$. Conditional on $(\tsigma^{(l)})_{l\geq1}$, as we have seen above, $(\Theta_N^{(l)})_{l\geq1}$ is a Gaussian process of covariance structure given by $(Q_{l,l'})_{l,l'\geq1}$. We want to define a coupling of the measures $(\proba_N)_{N\geq1}$ such that the variables of $(\Theta_N^{(l)})_{l\geq1}$ converge in $L^2$ towards some $(\theta_l)_{l\geq1}$ defined below. 

For each $K\geq1$, let $L_{N,K}$ be the Cholesky decomposition of the matrix $\bar{Q}_{N,K}=(Q_{l,l'})_{1\leq l,l'\leq K}$. That is, $L_{N,K}$ is the unique lower triangular matrix of positive diagonal entries in $\R^{K\times K}$ such that
\begin{equation}\label{eq:overlap-matrix}
    \bar{Q}_{N,K} = L_{N,K} L_{N,K}^\top.
\end{equation}
Let $(\zeta_l)_{l\geq1}$ be a standard Gaussian process. For each $K\geq1$, let $\bar{\zeta}_K$ be the vector in $\R^K$ of coordinates $\zeta_1,\dots,\zeta_K$. We then note that, for all $K\geq1$ we have that
\begin{equation}\label{eq:marginal_thetas}
    (\Theta_N^{(1)},\dots,\Theta_N^{(K)}) \overset{d}{=} L_{N,K} \bar{\zeta}_K.
\end{equation}
As we see in Appendix \ref{app:cholesky}, if $K'\geq K$, then $L_{N,K}$ is equal to the $K$-th leading principal sub-matrix of $L_{N,K'}$. This implies, by the Kolmogorov Extension Theorem (whose consistency assumption holds because of Lemma \ref{lem:cholesky_consistency} in Appendix \ref{app:cholesky}), that there is a well defined process whose $K$-th finite marginals are given by $L_{N,K} \bar{\zeta}_K$. By the equality in distribution \eqref{eq:marginal_thetas}, for all $N\geq1$, we can couple $(\Theta_N^{(l)})_{l\geq1}$ and $(\zeta_l)_{l\geq1}$ such that, for all $K\geq1$
\begin{equation*}
    (\Theta_N^{(1)},\dots,\Theta_N^{(K)}) = L_{N,K} \bar{\zeta}_K.
\end{equation*}

Let $\bar{Q}_K\in\R^{K\times K}$ be a symmetric matrix with $1$'s in the diagonal elements and $q$'s off-diagonal. Moreover, as we see in Appendix \ref{app:cholesky}, if $\E\thermal{\|\bar{Q}_{N,K}-\bar{Q}_K \|_F^2}$ goes to $0$ as $N$ goes to infinity, then $\E\thermal{\|L_{N,K}-L_K \|_F^2}$ also converges to $0$; where $L_K$ is the Cholesky matrix of the $K$-th leading principal sub-matrix of $\bar{Q}_K$. This means that, when $\E\thermal{\|\bar{Q}_{N,K}-\bar{Q}_K \|_F^2}$ vanishes, $(\Theta_N^{(1)},\dots,\Theta_N^{(K)})$ converges in $L^2$ towards the random vector
\begin{equation*}
    (\theta_1,\dots,\theta_K) := L_K \bar{\zeta}_K.
\end{equation*}

\subsubsection{Proof of the theorem}

\begin{lemma}\label{lem:log-bound}
    Let $(X_i)_{i\geq1}$ be i.i.d. random variables with mean $\mu = \E[X_1] > 0$ and variance $\sigma^2 = \text{Var}(X_1) < \infty$. Also assume that $X_1 \geq 0$ a.s. and that $\E[X_1^{-1}] < \infty$. Write $\bar{X}_K=\frac{1}{K}\sum_{i=1}^K X_i$. Then, for every $K\in\N$,
    \begin{equation*}
        \E\left|\log\left(\frac{\mu}{\bar{X}_K}\right)\right| \leq \frac{2\sigma}{\sqrt{K}} \left( \frac{1}{\mu} + 2\sigma\sqrt{\E[X_1^{-1}]} \right).
    \end{equation*}
\end{lemma}
\begin{proof}
    Set $\Delta = \mu-\bar{X}_K$ and $A := \{\bar{X}_K\ge \mu/2\} = \{\Delta\le \mu/2\}$. On $A$ we have $t: = \Delta/\mu \in [0,1/2]$. As in the standard bound for the logarithm, $\log(1-t)\geq -2t$ and $\log(1-t)\leq -t$ on this range, hence $|\log(1-t)| \leq 2|t|$. Therefore,
    \begin{equation*}
        \E\left|\log\left(\frac{\mu}{\bar{X}_K}\right)\mathbb{I}_A\right| = \E\left|\log\left(1-\frac{\Delta}{\mu}\right)\mathbb{I}_A\right| \leq \frac{2}{\mu}\E[|\Delta| \,\mathbb{I}_A] \leq \frac{2\sigma}{\mu \sqrt{K}},
    \end{equation*}
    since $\mathrm{Var}(\bar{X}_K)=\sigma^2/K$. For the tail on $A^c = \{\bar{X}_K<\mu/2\}$ we will use that for $u\geq1$, $0 \leq \log u \leq 2\sqrt{u}$, hence $|\log u| = \log u \leq 2\sqrt u$. On $A^c$ we have $\mu/\bar{X}_K\geq 2$, so
    \begin{equation*}
        \left|\log\left(\frac{\mu}{\bar{X}_K}\right)\mathbb{I}_{A^c}\right| \leq 2\sqrt{\frac{\mu}{\bar{X}_K}}\mathbb{I}_{A^c}.
    \end{equation*}
    By Cauchy-Schwarz,
    \begin{equation*}
        \E\left|\log\left(\frac{\mu}{\bar{X}_K}\right)\mathbb{I}_{A^c}\right| \leq 2\sqrt{\mu\E[\bar{X}_K^{-1}]\proba(A^c)}.
    \end{equation*}
    Using convexity of $x\mapsto x^{-1}$ on $(0,\infty)$ and Jensen,
    \begin{equation*}
        \E[\bar{X}_K^{-1}] \leq \E[X_1^{-1}] < \infty.
    \end{equation*}
    Finally, Chebyshev’s inequality gives $\proba(A^c) = \proba(\Delta>\mu/2) \leq 4\mathrm{Var}(\bar{X}_K)/\mu^2 = 4\sigma^2/(\mu^2 K)$. Combining the bounds,
    \begin{equation*}
        \E\left|\log\left(\frac{\mu}{\bar{X}_K}\right)\mathbb{I}_{A^c}\right| \leq \frac{4\sigma^2}{\sqrt{K}} \sqrt{\E[X_1^{-1}]}.
    \end{equation*}
    Adding the contributions on $A$ and $A^c$ yields the claim.
\end{proof}

\begin{lemma}\label{lem:quad-conv-mom}
    Let $\dH$ be as in Lemma \ref{lem:cavity-decomposition}. Then, for every $t_1,t_2 \in \R$, there are constants $C', C > 0$ such that, for all $N\geq1$,
    \begin{equation*}
        C \leq \E\thermal{\exp\{t_1\dH(\sigma_{i_0},\Theta_N)+t_2\dH(\sigma_{i_0},\theta)\}}_c \leq C'.
    \end{equation*}
\end{lemma}
\begin{proof}
    Let us first prove the upper bound. By Cauchy-Schwarz it is enough to prove the result for
        \begin{equation*}
        \E\thermal{\exp\{t_1\dH(\sigma_{i_0},\Theta_N)\}}_c \quad \text{and} \quad \E\thermal{\exp\{t_2\dH(\sigma_{i_0},\theta)\}}_c
    \end{equation*}
    separately. For the first bound notice that, because $J_{i_0}$ and $\sigma^{(l)}$ are independent under the cavity measure, we have that
    \begin{equation*}
        \E\thermal{\exp\{t_1\dH(\sigma_{i_0},\Theta_N)\}}_c = \exp\{2 t_1^2 \beta^2 +t_1| h_{i_0}|\}
    \end{equation*}
    where we just integrated with respect to $J_{i_0}$. Here the right hand side is bounded because, by assumption, the coordinates of $\sigma$ and $h$ have bounded support. The upper bound for
    \begin{equation*}
        \E\thermal{\exp\{t_2\dH(\sigma_{i_0},\theta)\}}_c
    \end{equation*}
    follows in a similar way by integrating with respect to $\theta$ instead of $J_{i_0}$. For the lower bound, by Jensen's inequality, we have
    \begin{equation*}
        \E\thermal{\exp\{t_1\dH(\sigma_{i_0},\Theta_N)+t_2\dH(\sigma_{i_0},\theta)\}}_c \geq \exp\left\{t_1\E\thermal{\dH(\sigma_{i_0},\Theta_N)}_c+t_2\E\thermal{\dH(\sigma_{i_0},\theta)}_c\right\}.
    \end{equation*}
    And we can easily see that
    \begin{equation*}
        \E\thermal{\dH(\sigma_{i_0},\Theta_N)}_c \geq h \thermal{\sigma_{i_0}}_c \geq -|h|;
    \end{equation*}
    where we used that $\E\thermal{2\beta \sigma_{i_0} \Theta_N}_c = 2\beta \thermal{\sigma_{i_0}}_c \E\thermal{\Theta_N}_c = 0$ because $\E\thermal{\Theta_N}_c=0$. The result then follows by a similar bound for $\E\thermal{\dH(\sigma_{i_0},\theta)}_c$.
\end{proof}

\begin{lemma}\label{lem:L2-conv-thetas}
    Assume that $\beta < 1/4$. Then, there is some $C > 0$ such that, for all $K\geq1$,
    \begin{equation*}
        \sum_{l=1}^K \E\thermal{|\Theta^{(l)}_N-\theta_l|}_c \leq C \sqrt{\frac{K^3}{N}}.
    \end{equation*}
\end{lemma}
\begin{proof}
    By the coupling constructed in Section \ref{sec:coupling}, we may write
    \begin{equation*}
        (\Theta_N^{(1)},\dots,\Theta_N^{(K)})=L_{N,K}\bar\zeta_K
        \qquad\text{and}\qquad
        (\theta_1,\dots,\theta_K)=L_K\bar\zeta_K,
    \end{equation*}
    where $L_{N,K}L_{N,K}^{\top}=\bar Q_{N,K}$ and $L_KL_K^\top=\bar Q_K$. Hence,
    \begin{equation*}
        \sum_{l=1}^K |\Theta_N^{(l)}-\theta_l|^2 = \|(L_{N,K}-L_K)\bar\zeta_K\|^2.
    \end{equation*}
    Taking expectation with respect to $\bar\zeta_K$, conditionally on the replicas, gives
    \begin{equation*}
        \E_{\bar\zeta}\left[\|(L_{N,K}-L_K)\bar\zeta_K\|^2\right]=\|L_{N,K}-L_K\|_F^2.
    \end{equation*}
    Therefore,
    \begin{equation*}
        \sum_{l=1}^K \E\thermal{|\Theta_N^{(l)}-\theta_l|^2}_c = \E\thermal{\|L_{N,K}-L_K\|_F^2}_c.
    \end{equation*}
    By Lemma \ref{lem:cholesky_continuity}, since the limiting matrix $\bar Q_K$ is positive definite for fixed $K$ whenever $q<1$, and since $\bar Q_{N,K}$ is close to $\bar Q_K$ with high probability, there exists a constant $C>0$ such that
    \begin{equation*}
        \E\thermal{\|L_{N,K}-L_K\|_F^2}_c \leq C \E\thermal{\|\bar Q_{N,K}-\bar Q_K\|_F^2}_c.
    \end{equation*}
    Now,
    \begin{equation*}
        \|\bar Q_{N,K}-\bar Q_K\|_F^2 = \sum_{l=1}^K (Q_{l,l}-1)^2+\sum_{l\neq l'}(Q_{l,l'}-q)^2.
    \end{equation*}
    Since $Q_{l,l}=(N-1)/N$, the diagonal contribution is $K/N^2$. For the off-diagonal terms, Proposition \ref{prop:rs} and exchangeability give
    \begin{equation*}
        \E\thermal{(Q_{l,l'}-q)^2}_c \leq \frac{C}{N}.
    \end{equation*}
    Hence,
    \begin{equation*}
        \E\thermal{\|\bar Q_{N,K}-\bar Q_K\|_F^2}_c \leq \frac{C K^2}{N}.
    \end{equation*}
    Combining the previous estimates proves
    \begin{equation*}
        \sum_{l=1}^K \E\thermal{|\Theta_N^{(l)}-\theta_l|^2}_c \leq \frac{C K^2}{N}.
    \end{equation*}
    By exchangeability, then
    \begin{equation*}
        \E\thermal{|\Theta_N^{(l)}-\theta_l|^2}_c = \E\thermal{|\Theta_N^{(1)}-\theta_1|^2}_c \leq \frac{C K}{N}.
    \end{equation*}    
    The final statement follows from Jensen's inequality to get
    \begin{equation*}
        \E\thermal{|\Theta_N^{(l)}-\theta_l|}_c \leq \sqrt{\E\thermal{|\Theta_N^{(l)}-\theta_l|^2}_c}.
    \end{equation*}
\end{proof}

\begin{lemma}\label{lem:L2-lln}
    For each $K\geq1$, let
    \begin{equation*}
        \bar{Z}_K := \frac{1}{K} \sum_{l=1}^K \exp\left\{\dH(\sigma^{(l)}_{i_0},\theta_l)\right\}.
    \end{equation*} 
    We then have that, there is some constant $C > 0$ such that, for all $K\geq1$,
    \begin{equation*}
        \E\thermal{\left|\log \bar{Z}_K - \log\thermal{\exp \{\dH(\sigma_{i_0},\theta)\}}_c\right|}_c \leq \frac{C}{\sqrt{K}}.
    \end{equation*}
\end{lemma}
\begin{proof}
    The result follows from combining Lemma \ref{lem:log-bound} with Lemma \ref{lem:quad-conv-mom}.
\end{proof}

We will now characterise the Radon-Nikodym derivative of $\proba_N$ with respect to the cavity measure $\proba_{N,c}$. This is done in the following lemma.

\begin{lemma}\label{lem:derivative}
    The random Radon-Nikodym derivative $d\proba_N/d\proba_{N,c}$, which for convenience we will denote $\mathcal{D}_N$, is given by
    \begin{equation*}
        \mathcal{D}_N = \frac{\exp\{\dH(\sigma_{i_0},\Theta_N)\}}{\thermal{\exp\{\dH(\sigma_{i_0},\Theta_N)\}}_c}.
    \end{equation*}
\end{lemma}
\begin{proof}
    Let $f:\R^N\mapsto\R$ be some integrable function with respect to $\proba(\cdot\mid h,J)$. We then have that
    \begin{equation*}
        \begin{split}
        \thermal{f(\sigma)} & = \frac{1}{Z_N} \sum_{\sigma\in\{-1,1\}^N} f(\sigma)\exp\{H_N(\sigma,h,J)\}  \\
            & = \frac{\sum_{\sigma\in\{-1,1\}^N} f(\sigma)\exp\{H_N(\sigma,h,J)\} }{\sum_{\sigma\in\{-1,1\}^N} \exp\{H_N(\sigma,h,J)\} } \\
            & = \frac{\sum_{\sigma\in\{-1,1\}^N} f(\sigma)\exp\{H_{N-1}(\tsigma,\tilh,\tJ)+\dH+\delta_N\} }{\sum_{\sigma\in\{-1,1\}^N} \exp\{H_{N-1}(\tsigma,\tilh,\tJ)+\dH+\delta_N\} } \\
            & = \frac{\exp\{\delta_N\}}{\exp\{\delta_N\}} \cdot \frac{\frac{1}{Z_{N,c}}\sum_{\sigma\in\{-1,1\}^N} f(\sigma)\exp\{H_{N-1}(\tsigma,\tilh,\tJ)+\dH\} }{\frac{1}{Z_{N,c}}\sum_{\sigma\in\{-1,1\}^N} \exp\{H_{N-1}(\tsigma,\tilh,\tJ)+\dH\} } \\
            & = \frac{\thermal{f(\sigma)\exp\{\dH\}}_c}{\thermal{\exp\{\dH\}}_c} = \thermal{f(\sigma) \frac{\exp\{\dH\}}{\thermal{\exp\{\dH\}}_c}}_c = \thermal{f(\sigma)\mathcal{D}_N}_c.
        \end{split}
    \end{equation*}
    Because this holds for every integrable $f$, we have then proved the lemma.
\end{proof}

\begin{proof}[Proof of Proposition \ref{prop:almost-main}]

    Let $\tilde\proba_N(\cdot|J)$ be defined according to the derivative
    \begin{equation*}
        \tilde{\mathcal{D}}_N := \frac{d\tilde\proba_N}{d\proba_{N,c}} = \frac{\int G(d\xi) \exp\{\dH(s,\theta)\}}{\int G(d\xi')\mu(ds') \exp\{\dH(s',\theta')\}},
    \end{equation*}
    with $G(\cdot)$ the standard Gaussian measure. From \eqref{eq:limit-marginal}, it is clear that $\tilde\proba_N$ is essentially the cavity measure $\proba_{N,c}$ but replacing the $i_0$-th marginal by $\mu_{i_0}(\cdot)$. Using the above lemmas and definitions we can now prove our main result. As in Section \ref{sec:general-strategy}, let
    \begin{equation*}
        I_{i_0} := I(\proba'_N,\tilde\proba_N) = \thermal{\mathcal{D}_N^{1/2} \ {\tilde{\mathcal{D}}}_N^{1/2}}_c
    \end{equation*}
    be the random Hellinger integral between $\proba'_N(\cdot|J)$ and $\tilde\proba_N(\cdot|J)$. Because the $i_0$-th marginal of $\tilde\proba_N$ is, for all $N\geq1$, equal to $\mu_{i_0}$, proving the appropriate Total Variation bounds between these two measures establishes the result. As discussed in Section \ref{sec:general-strategy}, to prove this, it is enough to bound $\E|I_{i_0}-1|$. But notice that, by Cauchy-Schwarz, we can see that $I_{i_0} \leq 1$. And also, $\{I_{i_0}=0\}$ has null probability. Furthermore, for all $x\in(0,1]$, we have that $|\log(x)| \geq |x-1|$. Then, we have
    \begin{equation*}
        \E|I_{i_0}-1| \leq  \E |\log I_{i_0}|.
    \end{equation*}
    Thus, it is enough to bound $\E |\log I_{i_0}|$ instead.
    
    For a cleaner notation, in the rest of the proof, the expectation $\thermal{\cdot}_c$ will also include expectation with respect to the independent standard Gaussian variables $\xi_1,\xi_2,...$. Taking this into account, notice that
    \begin{equation}\label{eq:log-hell-int}
        \begin{split}
            \log I_{i_0} & = \log \thermal{\exp\left\{\frac{1}{2}(\dH(\sigma_{i_0},\Theta_N) + \dH(\sigma_{i_0},\theta))\right\}}_c \\
            & \hspace{2cm}- \frac{1}{2}\log\thermal{\exp \{\dH(\sigma_{i_0},\Theta_N)\}}_c - \frac{1}{2}\log\thermal{\exp \{\dH(\sigma_{i_0},\theta)\}}_c.
        \end{split}
    \end{equation}
    For each $K\geq1$, let 
    \begin{equation*}
        \bar{X}_{N,K} := \frac{1}{K} \sum_{l=1}^K \exp\left\{\frac{1}{2}(\dH(\sigma^{(l)}_{i_0},\Theta^{(l)}_N) + \dH(\sigma^{(l)}_{i_0},\theta_l))\right\},
    \end{equation*}
    \begin{equation*}
        \bar{Y}_{N,K} := \frac{1}{K} \sum_{l=1}^K \exp\left\{\dH(\sigma^{(l)}_{i_0},\Theta^{(l)}_N)\right\}, \quad \text{and} \quad \bar{Z}_K := \frac{1}{K} \sum_{l=1}^K \exp\left\{\dH(\sigma^{(l)}_{i_0},\theta_l)\right\}.
    \end{equation*} 
    Define the following quantities
    \begin{equation*}
        \begin{dcases}
            \Delta_1 := \log \thermal{\exp\left\{\frac{1}{2}(\dH(\sigma_{i_0},\Theta_N) + \dH(\sigma_{i_0},\theta))\right\}}_c - \log \bar{X}_{N,K} \\
            \Delta_2 := \frac{1}{2}\log \bar{Y}_{N,K} - \frac{1}{2}\log\thermal{\exp \{\dH(\sigma_{i_0},\Theta_N)\}}_c \\
            \Delta_3 := \log \bar{X}_{N,K} - \log \bar{Z}_K \\
            \Delta_4 := \frac{1}{2} \log \bar{Z}_K - \frac{1}{2}\log \bar{Y}_{N,K} \\
            \Delta_5 := \frac{1}{2} \log \bar{Z}_K - \frac{1}{2}\log\thermal{\exp \{\dH(\sigma_{i_0},\theta)\}}_c.
        \end{dcases}
    \end{equation*}
    Clearly, $\log I_{i_0} = \sum_{k=1}^5 \Delta_k$. Then,
    \begin{equation}\label{eq:bound-logI}
        \E |\log I_{i_0}| \leq \sum_{k=1}^5 \E\thermal{|\Delta_k|}_c.
    \end{equation}
    We will now proceed to bound $\E\thermal{|\Delta_1|}_c,\dots,\E\thermal{|\Delta_5|}_c$ separately.
    
    By lemmas \ref{lem:log-bound} and \ref{lem:quad-conv-mom}, there are constants $C, C' > 0$ such that for all $N\geq 1$,
    \begin{equation*}
        \E\thermal{|\Delta_1|}_c \leq \frac{C}{\sqrt{K}} \quad \text{and} \quad \E\thermal{|\Delta_2|}_c \leq \frac{C'}{\sqrt{K}}.
    \end{equation*}
    Now, notice that, as functions of $\Theta_N^{(1)},\dots,\Theta_N^{(K)}$, $\log \bar{X}_{N,K}$ is $1$-Lipschitz. We then have that, for some $C'' >0$,
    \begin{equation*}
        \E\thermal{|\Delta_3|}_c \leq \sum_{l=1}^K \E\thermal{|\Theta_N^{(l)}-\theta_l|}_c\leq C'' \sqrt{\frac{K^3}{N}},
    \end{equation*}
    where we used Lemma \ref{lem:L2-conv-thetas} for the last inequality. By the same token we have that, for all $K\geq1$,
    \begin{equation*}
        \E\thermal{|\Delta_4|}_c \leq C'' \sqrt{\frac{K^3}{N}}.
    \end{equation*}
    Finally, Lemma \ref{lem:L2-lln} proves that, for some $C''' > 0$
    \begin{equation*}
        \E\thermal{|\Delta_5|}_c \leq \frac{C'''}{\sqrt{K}}.
    \end{equation*}
    If we set $K = N^{1/4}$, from equation \eqref{eq:bound-logI} and the bounds for $\E\thermal{|\Delta_1|}_c,\dots,\E\thermal{|\Delta_5|}_c$ we arrive at the conclusion of Theorem \ref{thm:main_sk}.

\end{proof}

Finally, Theorem \ref{thm:main_sk} then follows by combining Propositions \ref{prop:rs} and \ref{prop:almost-main}.
    
\paragraph{Acknowledgements.} M.S. would like to thank Pragya Sur for suggesting the cavity-contiguity connection and for the extensive discussions that followed. All possible errors are exclusive responsibility of the authors.

\appendix


\section{Properties of the Cholesky decomposition}\label{app:cholesky}

For each $K\geq 1$, let $Q\in \R^{K\times K}$ be a symmetric positive definite matrix and let $L(Q)$ be its Cholesky factor. That is, $L(Q)\in\R^{K\times K}$ is the unique lower triangular matrix with strictly positive diagonal entries such that
\begin{equation*}
    Q = L(Q) L^\top(Q).
\end{equation*}
The uniqueness property of such a factorisation is proved in \cite[Theorem 4.2.7]{golub2013matrix}. In this brief appendix, we will prove two relevant properties of Cholesky matrices that are used in Section \ref{sec:coupling} to define a coupling of the probability measures $\proba_N(\cdot)$ for different values of $N\geq1$.

The first of these properties is a consistency relation among Cholesky matrices of leading principal sub-matrices.

\begin{lemma}\label{lem:cholesky_consistency}
    Fix two integers $K' > K$, let $Q\in\R^{K'\times K'}$, and let $Q_K$ be its $K$-th leading principal sub-matrix. Then, $L(Q_K)$ is the $K$-th leading principal sub-matrix of $L(Q)$.
\end{lemma}
\begin{proof}
    This follows from \cite[Appendix B.1]{osborne2010bayesian}. Indeed, specialising the result presented there, we see that the Cholesky factor $L(Q)$ has $K$-th leading principal sub-matrix equal to $L(Q_K)$.
\end{proof}

Here we also prove a continuity property of the Cholesky decompositions of symmetric positive definite matrices. For this, consider the Frobenius norm of a matrix $A\in\R^{K\times K}$ given by
\begin{equation*}
    \|A\|_F := \sqrt{\sum_{i,j=1}^K A_{ij}^2}
\end{equation*}
The following lemma proves that, as a map on the set of symmetric and positive definite matrices of bounded spectrum, $L(\cdot)$ is Lipschitz with respect to the Frobenius norm.

\begin{lemma}\label{lem:cholesky_continuity}
    Assume there are $0 < a < b$ such that, for all $K\geq1$, the matrices $Q,Q'\in\R^{K\times K}$ are symmetric positive definite with eigenvalues in $[a,b]$. Then, 
    \begin{equation*}
        \|L(Q)-L(Q')\|_F \leq \frac{\sqrt{b}}{a}\|Q-Q'\|_F.
    \end{equation*}
\end{lemma}
\begin{proof}
    For all $t\in[0,1]$, define $Q_t := Q' + t(Q-Q')$. Since the set of symmetric matrices with spectrum contained in $[a,b]$ is convex, all matrices $Q_t$ have spectra contained in $[a,b]$. Write $Q_t = L_tL_t^\top$, with $L_t:=L(Q_t)$. By \cite[Theorem 1.3]{sun1991perturbation},
    \begin{equation*}
        \|\dot L_t\|_F \leq \kappa_2(Q_t)\frac{\|L_t\|_{\rm op}}{\|Q_t\|_{\rm op}}\|\dot Q_t\|_F,
    \end{equation*}
    where $\kappa_2(A):=\|A\|_{\rm op}\|A^{-1}\|_{\rm op}$ denotes the condition number of $A$. Since the spectrum of $Q_t$ is contained in $[a,b]$, we have $\kappa_2(Q_t)\leq b/a$, $\|L_t\|_{\rm op}=\sqrt{\|Q_t\|_{\rm op}}\leq \sqrt b$, and $\|Q_t\|_{\rm op}\geq a$. We then have
    \begin{equation*}
        \|\dot L_t\|_F \leq \frac{\sqrt b}{a}\|Q-Q'\|_F;
    \end{equation*}
    which implies
    \begin{equation*}
        \|L(Q)-L(Q')\|_F = \left\|\int_0^1 \dot L_t\,dt\right\|_F \leq \int_0^1 \|\dot L_t\|_F\,dt \leq \frac{\sqrt b}{a}\|Q-Q'\|_F.
    \end{equation*}
\end{proof}


\section{Proof of Proposition \ref{prop:rs}}\label{app:proof_applications}

We first recall a functional inequality established in \cite{bauerschmidt2019very} for the high-temperature regime of Ising models with positive semidefinite interaction matrices. For convenience, we formulate it directly for the two-replica system. For this, let $\Sigma := (\sigma^{(1)},\sigma^{(2)})\in\{-1,1\}^{2N}$ and consider the two-replica Gibbs measure on $\{-1,1\}^{2N}$ defined by
\begin{equation*}
    \proba^{(2)}_N(\Sigma|B,h) = \frac{1}{Z_N^2}\exp\left\{\Sigma^\top B\, \Sigma + l^\top \Sigma\right\}.
\end{equation*}
Here 
$B\in\R^{2n\times 2n}$ is the symmetric matrix given by
\begin{equation*}
    B = \gamma\begin{pmatrix}
        J & 0 \\
        0 & J
    \end{pmatrix}
\end{equation*}
and $l = (h, h)\in\mathbb R^{2N}$. As usual, we will denote expectation with respect to this measure by $\thermal{\cdot}$. In \cite{bauerschmidt2019very}, they prove that if $\|B\|_{\mathrm{op}}<2$, a log-Sobolev inequality holds. From this inequality, following standard arguments, a Poincaré-type inequality then follows. In particular, we have that if $\|B\|_{\mathrm{op}}<1/2$, there is some $C > 0$, such that for every test function $\phi:\{-1,1\}^{2N}\to\mathbb R$, we have
\begin{equation*}
    \thermal{(\phi(\Sigma)-\thermal{\phi(\Sigma)})^2} \leq \frac{C}{1-2\|B\|_{\mathrm{op}}}\sum_{i=1}^{2N}\thermal{\left|\partial_i \phi(\Sigma)\right|^2},
\end{equation*}
where $\partial_i \phi(\Sigma)$ is the discrete gradient on the hypercube. Furthermore, notice that, by definition $\|B\|_{\rm op} = \gamma\|J\|_{\rm op}$ and, by usual spectral concentration of Wigner matrices (see, for example \cite[Corollary 5.35]{vershynin2012introduction}), for every $\epsilon > 0$, we eventually have $\|J\|_{\rm op} < 2 + \epsilon$. We then have that, if $\gamma <1/4$, there is some $C >0$ such that
\begin{equation}\label{eq:poincare_sk}
    \thermal{(\phi(\Sigma)-\thermal{\phi(\Sigma)})^2} \leq C\sum_{i=1}^{2N}\thermal{\left|\partial_i \phi(\Sigma)\right|^2}.
\end{equation}

As we will see now, we can use this to prove the concentration of the overlaps. Indeed, if we let
\begin{equation*}
    R(\Sigma) := \sigma^{(1)\top}\sigma^{(2)},
\end{equation*}
then we have $|\partial_i R| \leq 2$. From \eqref{eq:poincare_sk} we conclude that, if $\gamma < 1/4$, there is a constant $C > 0$ such that
\begin{equation}\label{eq:thermal_conc_sk}
    \thermal{(Q_{12} - \thermal{Q_{12}})^2} \leq \frac{C}{N}.
\end{equation}
As stated in the following lemma, from this Poincaré-type inequality, we also get a bound for the variance of the overlaps.

\begin{lemma}\label{lem:overlap-concentration}
    If $\gamma < 1/4$, there exists a constant $C>0$ such that, for all $N\geq1$ large enough,
    \begin{equation*}
        \operatorname{Var}\left(\thermal{Q_{12}}\right)\leq \frac{C}{N}.
    \end{equation*}
\end{lemma}
\begin{proof}

    The proof of the lemma follows the strategy of \cite[Proposition 3.1.16]{talagrand2010mean} but replacing convex concentration inequalities with \eqref{eq:poincare_sk}. First, notice that $\operatorname{Var}(\thermal{Q_{12}}) = n^{-2} \operatorname{Var}(\thermal{R(\Sigma)})$ and consider $\thermal{R(\Sigma)}$ as a function of the Gaussian disorder. For simplicity we write the off-diagonal disorder as $J_{ij}=z_{ij}/\sqrt n$, with $(z_{ij})_{i<j}$ independent standard Gaussian variables. By the Gauss-Poincaré inequality \cite{boucheron2003concentration},
    \begin{equation}\label{eq:gauss_poin}
        \operatorname{Var}\left(\thermal{R(\Sigma)}\right)\leq \sum_{1\leq i<j\leq N}\E\left[\left(\frac{\partial}{\partial z_{ij}}\thermal{R(\Sigma)}\right)^2\right] = \E\left\|\nabla_z\thermal{R(\Sigma)}\right\|^2 .
    \end{equation}
    Differentiating the Gibbs expectation gives
    \begin{equation*}
        \frac{\partial}{\partial z_{ij}}\thermal{R(\Sigma)}=\frac{\gamma}{\sqrt n}\thermal{\left(R(\Sigma)-\thermal{R(\Sigma)}\right)\left(U_{ij}(\Sigma)-\thermal{U_{ij}(\Sigma)}\right)},
    \end{equation*}
    where $U:\{-1,1\}^{2N}\mapsto\R^{N \times N}$ is defined according to
    \begin{equation*}
        U_{ij}(\Sigma) := \sigma_i^{(1)}\sigma_j^{(1)}+\sigma_i^{(2)}\sigma_j^{(2)}.
    \end{equation*}
    Therefore,
    \begin{equation*}
        \operatorname{Var}\left(\thermal{R(\Sigma)}\right)\leq \frac{\beta^2}{n}\E\left\|\thermal{\left(R(\Sigma)-\thermal{R(\Sigma)}\right)\left(U(\Sigma)-\thermal{U(\Sigma)}\right)}\right\|_F^2
    \end{equation*}
    with $\|\cdot\|_F$ the Frobenius norm.
    
    Fix an arbitrary matrix $M\in\R^{N \times N}$ with $\|M\|_{\mathrm F}=1$ and let
    \begin{equation*}
        \phi_M(\Sigma) := \sigma^{(1)\top} M \sigma^{(1)} + \sigma^{(2)\top} M \sigma^{(2)}.
    \end{equation*}
    Notice that
    \begin{equation*}
        \thermal{(R(\Sigma)-\thermal{R(\Sigma)})(\phi_M(\Sigma)-\thermal{\phi_M(\Sigma)})} = \sum_{i,j=1}^N M_{ij}\thermal{(R(\Sigma)-\thermal{R(\Sigma)})(U_{ij}(\Sigma)-\thermal{U_{ij}(\Sigma)})}.
    \end{equation*}
    Therefore, by duality between the Frobenius norm and the Frobenius inner product,
    \begin{equation*}
        \sup_{\|M\|_{\mathrm F}=1}\thermal{(R(\Sigma)-\thermal{R(\Sigma)})(\phi_M(\Sigma)-\thermal{\phi_M(\Sigma)})}^2 = \left\|\thermal{\left(R(\Sigma)-\thermal{R(\Sigma)}\right)\left(U(\Sigma)-\thermal{U(\Sigma)}\right)}\right\|^2_{\mathrm F}.
    \end{equation*}

    We now estimate the variance of $\phi_M$. For every $1\leq k\leq N$, flipping the coordinate $\sigma_k^{(1)}$ changes $\phi_M$ by at most
    \begin{equation*}
        |\partial_k \phi_M(\Sigma)| \leq 4\left(|(M\sigma^{(1)})_k| + |(M^\top\sigma^{(1)})_k|\right).
    \end{equation*}
    Analogously, for $n+1\leq k\leq 2n$,
    \begin{equation*}
        |\partial_k \phi_M(\Sigma)| \leq 4\left(|(M\sigma^{(2)})_{k-n}| + |(M^\top\sigma^{(2)})_{k-n}|\right).
    \end{equation*}
    Using $(a+b)^2\leq 2a^2+2b^2$, we then obtain
    \begin{equation*}
        \sum_{k=1}^{2N}\thermal{|\partial_k \phi_M(\Sigma)|^2} \leq 32\thermal{\|M\sigma^{(1)}\|^2 + \|M^\top\sigma^{(1)}\|^2 + \|M\sigma^{(2)}\|^2 + \|M^\top\sigma^{(2)}\|^2}.
    \end{equation*}
    Since $\sigma^{(1)},\sigma^{(2)}\in\{-1,1\}^N$ and $\|M\|_{\mathrm{op}}\leq \|M\|_{\mathrm F}=1$, we have
    \begin{equation*}
        \|M\sigma^{(a)}\|^2 \leq \|M\|_{\mathrm{op}}^2\|\sigma^{(a)}\|^2 \leq N
    \end{equation*}
    for $a\in\{1,2\}$, and similarly for $M^\top$. Therefore,
    \begin{equation*}
        \sum_{k=1}^{2N}\thermal{|\partial_k \phi_M(\Sigma)|^2} \leq K n
    \end{equation*}
    for some constant $K>0$. By \eqref{eq:poincare_sk}, it follows that
    \begin{equation*}
        \thermal{(\phi_M(\Sigma)-\thermal{\phi_M(\Sigma)})^2} \leq K n.
    \end{equation*}

    Combining this with \eqref{eq:thermal_conc_sk} and using Cauchy-Schwarz,
    \begin{equation*}
        \thermal{(R(\Sigma)-\thermal{R(\Sigma)})(\phi_M(\Sigma)-\thermal{\phi_M(\Sigma)})}^2 \leq \thermal{(R(\Sigma)-\thermal{R(\Sigma)})^2}\thermal{(\phi_M(\Sigma)-\thermal{\phi_M(\Sigma)})^2} \leq K' n^2,
    \end{equation*}
    for some fixed $K' >0$. Since this bound holds uniformly over all matrices $M$ such that $\|M\|_{\mathrm F}=1$, we conclude that
    \begin{equation*}
        \left\|\thermal{\left(R(\Sigma)-\thermal{R(\Sigma)}\right)\left(U(\Sigma)-\thermal{U(\Sigma)}\right)}\right\|_{\mathrm F}^2 \leq K' n^2.
    \end{equation*}
    We then have
    \begin{equation*}
        \operatorname{Var}\left(\thermal{Q_{12}}\right) = \frac{1}{n^2} \operatorname{Var}\left(\thermal{R(\Sigma)}\right) \leq \frac{\beta^2}{n^3}\E\left\|\thermal{\left(R(\Sigma)-\thermal{R(\Sigma)}\right)\left(U(\Sigma)-\thermal{U(\Sigma)}\right)}\right\|_{\mathrm F}^2 \leq \frac{\beta^2K'}{n}.
    \end{equation*}

\end{proof}

For completeness, here we reproduce the proof from \cite{talagrand2010mean} for the uniqueness of the solutions of the fixed point equation for the SK model in high-temperature.

\begin{lemma}\label{lem:unique-sk-fixed-point}
    Assume that $\gamma < 1/2$. Then the fixed point equation
    \begin{equation*}
        q = \E\left[\tanh^2\left(2\beta\sqrt q\,Z+h\right)\right]
    \end{equation*}
    has a unique solution in $[0,1]$.
\end{lemma}

\begin{proof}

    Let $F:[0,1]\mapsto[0,1]$ be defined by
    \begin{equation*}
        F(q) = \E\left[\tanh^2\left(2\beta\sqrt q\,Z+h\right)\right].
    \end{equation*}
    Existence of solutions follows from continuity, since $F$ maps $[0,1]$ into itself. So we only need to prove uniqueness. Let $f(x) = \tanh^2(x)$. For $q>0$, using the Gaussian integration by parts formula,
    \begin{equation*}
        \frac{d}{dq}\E\left[f\left(2\beta\sqrt q\,Z+h\right)\right]= 2\beta^2\E\left[f''\left(2\beta\sqrt q\,Z+h\right)\right].
    \end{equation*}
    A direct computation gives
    \begin{equation*}
        f''(x) = 2\operatorname{sech}^4(x)-4\tanh^2(x)\operatorname{sech}^2(x).
    \end{equation*}
    Since $\operatorname{sech}^2(x)\leq 1$, we have $f''(x)\leq 2$ and then
    \begin{equation*}
        F'(q)\leq 4\beta^2\E\left[\operatorname{sech}^2\left(2\beta\sqrt q\,Z+h\right)\right] \leq 4 \beta^2.
    \end{equation*}
    If $\beta < 1/2$, this implies that $F$ is a contraction and the conclusion follows by the fixed point theorem.
    
\end{proof}

\begin{lemma}\label{lem:cavity-full-overlap-mean}
    Assume that $\beta<1/4$. Let $Q_{12}=N^{-1}\sum_{i=1}^N \sigma_i^{(1)}\sigma_i^{(2)}$ and $Q_{12}^{c}=N^{-1}\sum_{i\neq i_0}\sigma_i^{(1)}\sigma_i^{(2)}$. Then there exists $C>0$ such that
    \begin{equation*}
        \left|\E\thermal{Q_{12}}_N-\E\thermal{Q_{12}^{c}}_{N,c}\right|\leq \frac{C}{\sqrt{N}}.
    \end{equation*}
    Furthermore, there exists $C'>0$ such that
    \begin{equation*}
        \left|\E\thermal{Q_{12}}_{N-1}-\E\thermal{Q_{12}^{c}}_{N,c}\right|\leq \frac{C'}{\sqrt{N}}.
    \end{equation*}
\end{lemma}

As a final ingredient, we need to show that the means of the overlaps with respect to cavity and full differ in a small amount.

\begin{proof}
    Since $\left|Q_{12}-Q_{12}^{c}\right|\leq \frac{1}{N}$, it is enough to prove that
    \begin{equation*}
        \left|\E\thermal{Q_{12}^{c}}-\E\thermal{Q_{12}^{c}}_c\right|\leq \frac{C}{\sqrt{N}}.
    \end{equation*}
    For $t\in[0,1]$, define
    \begin{equation*}
        \thermal{F}_t := \frac{\thermal{F\exp\{tV_N\}}_c}{\thermal{\exp\{tV_N\}}_c},
    \end{equation*}
    where $V_N:=\Delta H_N(\sigma^{(1)})+\Delta H_N(\sigma^{(2)})$. Then $\thermal{\cdot}_0=\thermal{\cdot}_c$ and $\thermal{\cdot}_1=\thermal{\cdot}$. Differentiating,
    \begin{equation*}
        \frac{d}{dt}\thermal{Q_{12}^{c}}_t = \thermal{Q_{12}^{c}V_N}_t-\thermal{Q_{12}^{c}}_t\thermal{V_N}_t.
    \end{equation*}
    By Cauchy-Schwarz,
    \begin{equation*}
        \left|\frac{d}{dt}\thermal{Q_{12}^{c}}_t\right| \leq \thermal{\left(Q_{12}^{c}-\thermal{Q_{12}^{c}}_t\right)^2}_t^{1/2}\thermal{\left(V_N-\thermal{V_N}_t\right)^2}_t^{1/2}.
    \end{equation*}
    The Poincaré inequality from above implies
    \begin{equation*}
        \thermal{\left(Q_{12}^{c}-\thermal{Q_{12}^{c}}_t\right)^2}_t \leq \frac{C}{N},
    \end{equation*}
    while $\E\thermal{(V_N-\thermal{V_N}_t)^2}_t\leq C$. Hence,
    \begin{equation*}
        \left|\E\frac{d}{dt}\thermal{Q_{12}^{c}}_t\right|\leq \frac{C}{\sqrt{N}}.
    \end{equation*}
    Integrating over $t\in[0,1]$ yields
    \begin{equation*}
        \left|\E\thermal{Q_{12}^{c}}-\E\thermal{Q_{12}^{c}}_c\right|\leq \frac{C}{\sqrt{N}},
    \end{equation*}
    which finishes the first part of the lemma. The second one follows in an analogous manner.
\end{proof}

Notice that this lemma along with Lemma \ref{lem:overlap-concentration} prove that, if $\E\thermal{Q_{12}}$ converges to a constant $q$, then
\begin{equation*}
    \E\thermal{(Q_{12}-q)^2} \leq \frac{C}{N}.
\end{equation*}
We are then only left to prove that $\E\thermal{Q_{12}}$ converges. The lemma also implies that, when they exist, we can interchange the limits of $\E\thermal{Q_{12}}_c$ and $\E\thermal{Q_{12}}$ whenever necessary. We will do so in the following argument.

Let $q$ be the solution of $q=F(q)$ as in the above lemma. To finish the proof of Proposition \ref{prop:rs}, it is enough to note that $\E\thermal{Q_{12}}$ is a bounded sequence. Then, for every subsequence, there is a sub-subsequence such that it converges to some constant. Now, assume that there some of these sub-subsequences such that it converges to a constant $q'$ different from $q$. Over this sub-subsequence, Proposition \ref{prop:almost-main} holds and thus we have, by a straightforward adaptation of the proof of Corollary \ref{cor:local-observables}, that
\begin{equation*}
    q' = \E\left[\tanh^2\left(2\beta\sqrt q'\,Z+h\right)\right].
\end{equation*}
But by Lemma \ref{lem:unique-sk-fixed-point}, because we are assuming that $\beta < 1/4$, there is a unique solution to this equation and it is equal to $q$. This is a contradiction. Therefore, for every subsequence, there is a sub-subsequence such that $\E\thermal{Q_{12}}$ converges to $q$. We then conclude that
\begin{equation*}
    \lim_{N\to\infty} \E\thermal{Q_{12}} = q.
\end{equation*}
This along with Lemma \ref{lem:overlap-concentration} prove Proposition \ref{prop:rs}.

\end{document}